\documentclass{ifacconf}
\let\theoremstyle\relax
\usepackage{amsmath,amssymb,amsfonts}
\usepackage{mathtools}
\usepackage{graphicx}      
\usepackage{framed,xcolor}
\usepackage{color}
\usepackage{natbib}      
\usepackage{mathtools,amsthm}
 
\newtheorem{theorem}{Theorem}

\newtheorem{Lemma}{Lemma}
\newtheorem{Definition}{Definition}

\newtheorem{Assumption}{Assumption}
\newcommand{\vect}[1]{\ensuremath{\boldsymbol{\mathrm{#1}}}}
\definecolor{arash}{rgb}{0.8,0.8,1}
\definecolor{seb}{rgb}{0.8,1,0.8}
\definecolor{seb2}{rgb}{0.5,.5,1}
\definecolor{try}{rgb}{0.99,0.93,0.0}

\graphicspath{{Figures/}}

\usepackage[algo2e,norelsize,ruled,vlined,linesnumbered]{algorithm2e}
\begin{document}
\begin{frontmatter}
\title{Verification of Dissipativity and Evaluation of Storage Function
in Economic Nonlinear MPC using Q-Learning} 
\author[First]{Arash Bahari Kordabad} 
\author[First]{Sebastien Gros} 
\address[First]{Department of Engineering Cybernetics, Norwegian University of Science and Technology (NTNU), Trondheim, Norway (e-mail: \{Arash.b.kordabad, sebastien.gros\}@ntnu.no).}
\begin{abstract}                
In the Economic Nonlinear Model Predictive (ENMPC) context, closed-loop stability relates to the existence of a storage function satisfying a dissipation inequality. Finding the storage function is in general-- for nonlinear dynamics and cost-- challenging, and has attracted attentions recently. Q-Learning is a well-known Reinforcement Learning (RL) techniques that attempts to capture action-value functions based on the state-input transitions and stage cost of the system. In this paper, we present the use of the Q-Learning approach to obtain the storage function and verify the dissipativity for discrete-time systems subject to state-input constraints. We show that undiscounted Q-learning is able to capture the storage function for dissipative problems when the parameterization is rich enough. The efficiency of the proposed method will be illustrated in the different case studies.
\end{abstract}
\end{frontmatter}
\section{Introduction}
\par  Economic Nonlinear Model Predictive Control (ENMPC) optimizes an economic performance rather than penalizing deviations from a desired steady-state reference (\cite{rawlings2009optimizing}). Dissipativity is a fundamental concept in ENMPC, when the objective stage cost is not necessarily positive definite  (\cite{diehl2010lyapunov}). In contrast to tracking costs, an optimal economic policy may not lead to the closed-loop stability of the optimal steady-state point (\cite{amrit2011economic}). The dissipativity property allows one to equalize the ENMPC with a tracking MPC which has a well-established stability conditions. Unfortunately, it is not trivial to prove that a given problem is dissipative (\cite{pirkelmann2019approximate}). In order to prove it, one has to find a storage function that satisfies the dissipation inequality.~\cite{scherer2015linear} used Linear Matrix Inequality (LMI) techniques to compute the storage function for linear systems with a quadratic stage cost. A similar method is used by \cite{koch2020provably} to verify dissipativity properties based on noisy data for linear systems. Furthermore, the Sum-of-Squares (SOS) method is used by ~\cite{pirkelmann2019approximate} for polynomials dynamics and stage cost. \cite{zanon2016tracking} showed how tracking MPC schemes can approximate generic economic NMPC locally and provided conditions for the local existence of a storage function using a Semi-Definite Program (SDP). \cite{grune2016relation} showed that strict dissipativity with the addition of a suitable controllability property leads the turnpike property in optimal control. Roughly speaking, the turnpike property means that open-loop optimal trajectories spend \textit{most of the time} in a vicinity of the steady-state solution.
\par Reinforcement Learning (RL) offers tools for tackling Markov Decision Processes (MDP) without having an accurate knowledge of the probability distribution underlying the state transition (\cite{sutton2018reinforcement,bertsekas2019reinforcement}). RL seeks to optimize the parameters underlying a given policy in view of minimizing the expected sum of a given stage cost. Q-learning attempts to find the optimal policy by capturing the optimal action-value function underlying the MDP. Model Predictive Control (MPC) scheme used as a value function approximator, provides a formal framework to analyse the stability and feasibility of the closed-loop system (\cite{MPCbook}). \cite{gros2020reinforcement,gros2019data} showed that adjusting the MPC model, cost and constraints allows the MPC to capture the optimal policy for the system even if using an inaccurate model of the system. RL is proposed and recent research have focused on MPC-based approximation for RL (see e.g. ~\cite{koller2018learning,mannucci2017safe,gros2020reinforcement,bahari2021reinforcement,kordabad2021mpc,nejatbakhsh2021reinforcement}).
\par In this paper, we propose to use undiscounted Q-learning as a method to compute the storage function and verify the dissipativity conditions. We assume that the given problem is dissipative and replace the ENMPC by its equivalent tracking MPC-scheme (positive stage cost) with an additional storage function. We parametrized the storage function, stage cost and terminal cost in the tracking MPC-scheme, and use it as a function approximator in Q-learning based on the economic stage cost. Then Q-learning adjusts the parameters by capturing the optimal action-value function. Using different examples we show that this method can be used for more general, nonlinear dynamics and cost to deliver the storage function with high accuracy. 
\par The rest of the paper is structured as follows. Section \ref{sec:ENMPC} provides the ENMPC formulation and the dissipativity conditions and explains how one can parameterize the ENMPC-scheme. Section \ref{sec:Q:learning} details the Q-learning approach and constrained learning steps and discusses efficiency of the proposed method. Section \ref{sec:sim} illustrates the simulations results for the different case studies. Section \ref{sec:conc} delivers a conclusion.
\section{Economic Nonlinear MPC}\label{sec:ENMPC}
In this section we detail the ENMPC formulation and the dissipativity condition. Consider the following discrete time, constrained dynamic system:
\begin{align}\label{eq:dyn}
    \vect s_{k+1}= \vect f\left(\vect s_k,\vect a_k\right), \quad \vect h \left(\vect s_k,\vect a_k\right) \leq 0,
\end{align}
where $k=0,1,\ldots$ is the physical time index, $\vect s_k\in\mathbb{X}\subseteq\mathbb{R}^n$ is the state, $\vect a_k\in\mathbb{U}\subseteq\mathbb{R}^m$ is the input. Vector field $ \vect f: \mathbb{R}^{n+m}\rightarrow \mathbb{R}^n$ expresses the state transition. Function $\vect h$ is a mixed input-state constraint. The set of feasible state-input pairs is defined as follows:
\begin{align}
    \mathbb{Z}:=\{(\vect s,\vect a)\in \mathbb{X}\times \mathbb{U}|\vect f\left(\vect s,\vect a\right)\in\mathbb{X},\,\vect h \left(\vect s,\vect a\right) \leq 0 \}
\end{align}
The selected economic stage cost, denoted by $L:\mathbb{Z}\rightarrow \mathbb{R}$, is not necessarily positive definite. The following standard assumption is essential in ENMPC context, that we use in the rest of the paper.
\begin{Assumption}\label{assum1}
The set $\mathbb{Z}$ is non-empty and compact and the  cost $L(\cdot)$ and function $\vect f(\cdot)$ are continuous on $\mathbb{Z}$.
\end{Assumption}
The corresponding optimal steady-state pair $(\vect s_e,\vect a_e)$ is defined as follows:
\begin{align} \label{eq:steady}
    (\vect s_e,\vect a_e)=\mathrm{arg}\,\min_{s,a}\,\, \{ L(\vect s, \vect a) |\,\, \vect s=\vect f(\vect s, \vect a),\,\,\vect h(\vect s, \vect a)\leq 0\}
\end{align}  
We define a \emph{shifted} stage cost $\ell$, as follows:
\begin{align}\label{eq:shift:l}
    \ell(\vect s,\vect a)= L(\vect s, \vect a)-L(\vect s_e, \vect a_e).
\end{align}
Then $\ell(\vect s_e,\vect a_e)=0$. A deterministic policy $\vect\pi: \mathbb{X}\rightarrow \mathbb{U}$ maps the state space to the input space. We seek an optimal policy $\vect\pi^\star$ solution of the following infinite-horizon problem:
\begin{subequations}\label{eq:per}
\begin{align}
    V^\star(\vect s)=\min_{\vect\pi}&\quad \sum_{k=0}^{\infty} \ell(\vect x_k,\vect\pi(\vect x_k)) \label{eq:cost}\\
\mathrm{s.t.}&\quad \vect x_{k+1} =  \vect f\left(\vect x_{k},\vect\pi(\vect x_k)\right),\,\,\, \vect x_{0} = \vect s, \label{eq:cons1}\\
&\quad\vect h \left(\vect{x}_{k},\vect\pi(\vect x_k)\right) \leq 0, \label{eq:cons2}
\end{align}
\end{subequations}
where $V^\star$ is the optimal value function and the sequence $(\vect x_{k})_{k=1}^{\infty}$ is the state trajectory under policy $\vect \pi$ starting from an arbitrary state $\vect x_0=\vect s$. The optimal action-value function is defined as follows:
\begin{align}\label{eq:Q:optimal}
Q^\star (\vect s,\vect a)=\ell(\vect s,\vect a)+V^\star(\vect f(\vect s,\vect a))
\end{align}
In the following we define the concept of \emph{dissipativity}, which is crucial in the ENMPC context.
\begin{Definition}
The system \eqref{eq:dyn} with stage cost $L$ is strictly dissipative if there exists a continuous \emph{storage} function $\lambda: \mathbb{X}\rightarrow\mathbb{R}$ satisfying:
\begin{align}\label{eq:dissip:l}
\lambda (\vect f(\vect s, \vect a)) - \lambda (\vect s) \leq & -\rho(\|\vect s-\vect s_e\|)+ \ell(\vect s, \vect a)
\end{align}
for all $(\vect s, \vect a)\in \mathbb{Z}$ and some function $\rho(\cdot)\in \mathcal{K}_\infty$, where $\|\cdot\|$ indicates an Euclidean norm.
\end{Definition}
Note that adding a constant in the storage function does not invalidate \eqref{eq:dissip:l}. Hence, we can assume that $\lambda(\vect s_e)=0$ without loss of generality. If the storage function $\lambda$ exists, the \emph{rotated} stage cost $\bar\ell$ is defined as follows:
\begin{align}\label{eq:rotated:L}
    \bar \ell(\vect s, \vect a)&=\ell(\vect s, \vect a)+\lambda(\vect s)-
    \lambda(\vect f(\vect s, \vect a))
\end{align}
From \eqref{eq:rotated:L} and \eqref{eq:dissip:l}, we have:
\begin{align}\label{eq:l:bar}
    \rho(\|\vect s-\vect s_e\|) \leq \bar \ell(\vect s, \vect a) , \quad \bar\ell(\vect s_e, \vect a_e)=0
\end{align}
Under assumption \ref{assum1}, the cost and storage function will remain bounded over the set $\mathbb{Z}$. Using a telescoping sum, the cost \eqref{eq:cost} can be expressed based on the positive stage cost $\bar\ell$ as follows:
\begin{align}\label{eq:cost:2}
    &\sum_{k=0}^{\infty} \ell(\vect x_k,\vect\pi(\vect u_k))=\sum_{k=0}^{\infty} \bar\ell(\vect x_k,\vect\pi(\vect x_k))-\lambda(\vect x_k)+\lambda(\vect x_{k+1})=\nonumber\\
    &-\lambda(\vect x_0)+\lambda(\vect x_{\infty})+  \sum_{k=0}^{\infty} \bar\ell(\vect x_k,\vect\pi(\vect x_k))=\nonumber\\&\qquad\qquad\qquad\qquad\qquad\quad-\lambda(\vect s)+\sum_{k=0}^{\infty} \bar\ell(\vect x_k,\vect\pi(\vect x_k))
\end{align}
Note that a dissipative ENMPC is equal to a tracking MPC with the rotated stage cost $\bar\ell$ which is zero at the optimal steady-state and lower bounded by a $\mathcal{K}_{\infty}$ function (see \eqref{eq:l:bar}) and the closed-loop dynamics with optimal policy will be stable (see \cite{MPCbook}). Hence, the state trajectory converges to its optimal steady-state, i.e. $\lim_{k\rightarrow\infty} \vect x_{k}=\vect s_e$, and $\lambda(\vect x_{\infty})=\lambda(\vect s_e)=0$. Hence, \eqref{eq:per} can be written as follows:
\begin{subequations}\label{eq:per2}
\begin{align}
    V^\star(\vect s)=\min_{\vect\pi}&\quad-\lambda(\vect s)+\sum_{k=0}^{\infty} \bar\ell(\vect x_k,\vect\pi(\vect x_k)) \label{eq:cost:MPC} \\
\mathrm{s.t.}&\quad \eqref{eq:cons1}, \eqref{eq:cons2}
\end{align}
\end{subequations}
From \eqref{eq:l:bar}, the stage cost of \eqref{eq:per2} is lower bounded by a $\mathcal{K}_\infty$ function and the term $-\lambda(\vect s)$ in the cost \eqref{eq:cost:MPC} has no effect on the optimal policy, because it only depends on the current state $\vect s$. Hence, the tracking NMPC \eqref{eq:per2} delivers the same input/state solution and same policy as ENMPC\eqref{eq:per} and the closed-loop behaviour of \eqref{eq:per} and \eqref{eq:per2} are equivalent. 
Finding the storage function $\lambda(\vect s)$ that satisfies \eqref{eq:dissip:l} is not trivial. In order to evaluate the storage function, we use the following finite-horizon MPC-scheme parameterized by $\vect \theta$:
\begin{subequations}\label{eq:V}
\begin{align}
V_{\vect\theta} (\vect s)=\min_{\vect{u},\vect x}&\quad-\lambda_{\vect\theta}(\vect s)+T_{\vect\theta}(\vect x_{N})+\sum_{k=0}^{N-1} \hat\ell_{\vect\theta}(\vect x_{k},\vect u_{k})\label{eq:MPC:cost}\\
\mathrm{s.t.}&\quad \vect x_{k+1} = \vect f\left(\vect x_{k},\vect u_{k}\right),\,\,\, \vect x_{0} = \vect s, \label{eq:Dynamics:MPC}\\
&\quad\vect h \left(\vect{x}_{k},\vect{u}_{k}\right) \leq 0, \label{eq:SafetyConstraints:MPC}
\end{align}
\end{subequations}
where $\lambda_{\vect\theta}$ is the approximated storage function, $\hat\ell_{\vect\theta}$ is the stage cost and $T_{\vect\theta}$ is the terminal cost penalizing the finite-horizon problem. Vector $\vect x=[\vect x^\top_0,\ldots,\vect x^\top_N]^\top$ is the predicted state trajectory, $\vect u=[\vect u^\top_0,\ldots,\vect u^\top_N]^\top$ is the input profile. State $\vect s$ is the current state of the system and $N$ is the horizon length. Value function $V_{\vect\theta}$ estimates the optimal value
function $V^\star$. Using \eqref{eq:rotated:L}, the parametrized rotated cost $\bar \ell_{\vect\theta}$ is obtained as follows:
\begin{align}\label{eq:rotated:L:theta}
    \bar \ell_{\vect\theta}(\vect s, \vect a)=\ell(\vect s, \vect a)+\lambda_{\vect\theta}(\vect s)-
    \lambda_{\vect\theta}(\vect f(\vect s, \vect a))
\end{align}
Since we use a finite-horizon MPC-scheme \eqref{eq:V} to approximate the infinite-horizon problem \eqref{eq:per2}, in general $\hat \ell_{\vect\theta}$ should be selected differently than the rotated stage cost $\bar \ell_{\vect\theta}$. The parametrized rotated stage cost $\bar \ell_{\vect\theta}$ is directly obtained by the storage function using \eqref{eq:rotated:L:theta} but the terminal cost $T_{\vect\theta}$ and stage cost $\hat\ell_{\vect\theta}$ are selected as parametric functions. We force $\hat \ell_{\vect\theta}$ to be lower bounded by a $\mathcal{K}_\infty$ function in the Q-learning steps or by construction, but we do not impose any restriction on $\bar \ell_{\vect\theta}$ or the storage function $\lambda_{\vect\theta}$. If the parameterization is rich enough, Q-learning will be able to capture the optimal action-value function under some mild assumptions shown in \cite{zanon2020safe}. As a result, the stage cost, terminal cost and storage function can be adjusted to deliver the optimal action-value and optimal policy. As a result:
\begin{align}
    \rho(\|\vect s-\vect s_e\|) \leq \bar \ell_{\vect\theta}(\vect s, \vect a)
\end{align}
holds for some $\rho(.)\in \mathcal{K}_\infty$, and it implies the dissipativity inequality \eqref{eq:dissip:l}.

Next section details the conditions on the stage cost $\hat \ell_{\vect\theta}$ and discusses how $\vect\theta$ should be updated using undiscounted Q-learning to capture the storage function, optimal policy and optimal action-value function while respecting the possible constraints on the parameters $\vect\theta$.
\section{Evaluation of the Storage function using Q-learning}\label{sec:Q:learning}
\par Q-learning is a classical model-free RL algorithm that tries to capture the optimal action-value function via tuning the parameter vector $\vect\theta$. The MPC-based action-value function can be extracted as follows (see \cite{gros2019data}):
\begin{subequations}\label{eq:Q}
\begin{align}
Q_{\vect\theta} (\vect s,\vect a)=\min_{\vect{u},\vect x}&\quad \eqref{eq:MPC:cost}\\
\mathrm{s.t.}&\quad  \eqref{eq:Dynamics:MPC}, \eqref{eq:SafetyConstraints:MPC} \label{eq:cons1:Q}\\
&\quad \vect u_{0}=\vect a  \label{eq:cons2:Q}
\end{align}
\end{subequations}
The parameterized deterministic policy can be obtained as: 
 \begin{align}\label{eq:pi}
    \vect \pi_{\vect \theta}(\vect s)=\vect {\bar u}_{0}^{\star}(\vect s,\vect \theta),
\end{align}
where $\vect {\bar u}_{0}^{\star}(\vect s,\vect \theta)$ is the first element of ${\vect {\bar u}}^\star$, which is the input solution of the MPC scheme \eqref{eq:V}.
It is shown by~\cite{gros2019data} that the value function $V_{\vect\theta}$ in \eqref{eq:V}, the action-value function $Q_{\vect\theta}$ in \eqref{eq:Q} and the policy $\vect \pi_{\vect\theta}$ in \eqref{eq:pi} satisfy the fundamental Bellman equations. For an undiscounted problem, basic Q-learning uses the following update rule for the parameters $\vect\theta$ at state $\vect s_k$ (see e.g. ~\cite{sutton2018reinforcement}):
\begin{subequations}\label{eq:Qlearning}
	\begin{gather}
	\delta_{k}=\ell(\vect s_{k},\vect a_{k})+ V_{\vect \theta}(\vect s_{k+1})-Q_{\vect\theta}(\vect s_{k},\vect a_{k})\label{eq:TD1}\\
	\vect \theta_{k+1}= \vect \theta_{k}+\alpha\delta_{k}\nabla_{\vect \theta}Q_{\vect \theta}(\vect s_{k},\vect a_{k}) \label{Qlearning}
	\end{gather}
\end{subequations}
where the scalar $\alpha>0$ is the learning step-size, $\delta_k$ is labelled the Temporal-Difference (TD) error and
the input $\vect a_k$ is selected according to the corresponding parametric policy $\vect \pi_{\vect\theta}(\vect s_k)$.
\par Since the constraints in the MPC \eqref{eq:V} are independent of $\vect\theta$, the sensitivity $\nabla_{\vect\theta}Q_{\vect\theta}(\vect s)$ required in Eq.\eqref{Qlearning} only depends on the cost and it is given by (see ~\cite{gros2019data}):
\begin{gather}
    \nabla_{\vect\theta}Q_{\vect\theta}(\vect s,\vect a)=\nabla_{\vect \theta}\Phi_{\vect \theta}(\vect u^\star,\vect x^\star)
\end{gather}
where $\Phi_{\vect\theta}$ gathers the cost for \eqref{eq:MPC:cost} and $\vect x^\star, \vect u^\star$  is the solution of \eqref{eq:Q} for a given $\vect s, \vect a$ pair. 

In the following we made an assumption in parameterization that allows us to evaluate optimal action-value and the storage function. 
\begin{Assumption}\label{assum:rich}
We assume that the parametrization is rich enough. I.e. there exists $\vect\theta^\star$ such that:
\begin{subequations}\label{eq:rich}
\begin{align}
    \vect \pi_{\vect\theta^\star}(\vect s)&=\vect \pi^{\star}(\vect s),\quad V_{\vect\theta^\star}(\vect s)=V^{\star}(\vect s),\label{eq:rich:V}\\ Q_{\vect\theta^\star}(\vect s,\vect a)&=Q^{\star}(\vect s,\vect a)\label{eq:rich:Q}
\end{align}
\end{subequations}
\end{Assumption}
\begin{Lemma}
If the assumption \ref{assum:rich} holds, then $\vect \theta^\star$ is a fixed point of \eqref{eq:Qlearning}.
\end{Lemma}
\begin{proof}From \eqref{eq:Q:optimal} and \eqref{eq:rich}, we have:
\begin{align} 
   \delta&=\ell(\vect s,\vect a)+ V_{\vect \theta^\star}(\vect f(\vect s,\vect a))-Q_{\vect\theta^\star}(\vect s,\vect a)\nonumber\\&=\ell(\vect s,\vect a)+ V^\star(\vect f(\vect s,\vect a))-Q^\star(\vect s,\vect a)=0\qedhere
\end{align}
\end{proof}
Let us define $\Psi^{N}_{\vect\theta}(\vect s)$ and $\Lambda^N_{\vect\theta}(\vect s,\vect a)$ as follows:
\begin{subequations}\label{eq:phi:2}
\begin{align}
\Psi^N_{\vect\theta}(\vect s):=\min_{\vect{u},\vect x}&\quad T_{\vect\theta}(\vect x_{N})+\sum_{k=0}^{N-1} \hat\ell_{\vect\theta}(\vect x_{k},\vect u_{k})\nonumber\\
\mathrm{s.t.}& \quad  \eqref{eq:Dynamics:MPC}, \eqref{eq:SafetyConstraints:MPC}\label{eq:phi:2:1}\\
\Lambda^N_{\vect\theta}(\vect s,\vect a):=\min_{\vect{u},\vect x}&\quad T_{\vect\theta}(\vect x_{N})+\sum_{k=0}^{N-1} \hat\ell_{\vect\theta}(\vect x_{k},\vect u_{k})\nonumber\\
\mathrm{s.t.}& \quad  \eqref{eq:cons1:Q}, \eqref{eq:cons2:Q}\label{eq:phi:2:2}
\end{align}
\end{subequations}
One can observe that:
\begin{align}\label{eq:QV}
    \Lambda^N_{\vect\theta}(\vect s,\vect a)=\hat\ell_{\vect\theta}(\vect s,\vect a)+\Psi^{N-1}_{\vect\theta}(\vect f(\vect s,\vect a))
\end{align}
Before providing the main theorem that expresses how Q-learning can capture a valid storage function we make a basic stability assumption.
\begin{Assumption}\label{assum:T}
For all $\vect s\in \mathbb{X}$, there exists an input $\vect a$ (such that $(\vect s,\vect a)\in\mathbb{Z}$) satisfying:
\begin{align}\label{eq:T:assu}
 T_{\vect\theta}(\vect f(\vect s,\vect a))-T_{\vect\theta}(\vect s)\leq- \hat \ell_{\vect\theta}(\vect s,\vect a)   ,\quad \vect h(\vect s, \vect a)\leq 0
\end{align}
\end{Assumption}
Assumption \ref{assum:T} is a basic assumption in the MPC context that allows one to discuss about the stability. If there exists a terminal constraint in form $\vect x_N\in \mathbb{X}_f\subseteq \mathbb{X}$ in the MPC \eqref{eq:V}, then the assumption \ref{assum:T} will be limited to be satisfied for all $\vect s \in \mathbb{X}_f$. To satisfy that in the Q-learning context, one needs to provide a generic function approximator for the terminal cost $T_{\vect\theta}$. The following Lemma on the monotonicity property of the MPC value function is useful in the building of the theorem.
\begin{Lemma}
Under the assumptions \ref{assum1} and \ref{assum:T}, we have:
\begin{align}\label{eq:mono}
    \Psi^{N+1}_{\vect\theta}(\vect s) \leq \Psi^{N}_{\vect\theta}(\vect s) ,\quad \forall N\geq0
\end{align}
\end{Lemma}
\begin{proof}
See \cite{MPCbook}.
\end{proof}
\begin{theorem}\label{theorem}
Assumption \eqref{assum1}-\eqref{assum:T}, for a dissipative problem with rich enough parameterization, the undiscounted Q-learning yields a valid storage function which satisfies \eqref{eq:dissip:l} if there exists a $\mathcal{K}_\infty$ functions $\alpha_1(.)$ satisfying
        \begin{align}
       \alpha_1(\|\vect s-\vect s_e\|)&\leq \hat \ell_{\vect\theta^\star}(\vect s,\vect a),\quad \forall (\vect s,\vect a)\in\mathbb{Z}
    \end{align}
\end{theorem}
\begin{proof}
Using  \eqref{eq:rich:Q} and \eqref{eq:Q}, we have:
\begin{align}\label{eq:QV1}
    Q^\star(\vect s,\vect a)&=Q_{\vect\theta^\star}(\vect s,\vect a)
    =-\lambda_{\vect\theta^\star}(\vect s)+\Lambda^N_{\vect\theta^\star}(\vect s,\vect a) \nonumber\\
    &=-\lambda_{\vect\theta^\star}(\vect s)+\hat\ell_{\vect\theta^\star}(\vect s,\vect a)+\Psi^{N-1}_{\vect\theta^\star}(\vect f(\vect s,\vect a))
\end{align}
Moreover, from \eqref{eq:Q:optimal},\eqref{eq:rich:V} and \eqref{eq:V}, we have:
\begin{align}\label{eq:QV3}
    Q^\star(\vect s,\vect a)&=\ell(\vect s,\vect a)+ V^\star(\vect f(\vect s,\vect a))=\ell(\vect s,\vect a)+ V_{\vect\theta^\star}(\vect f(\vect s,\vect a))\nonumber\\&=\ell(\vect s,\vect a)-\lambda_{\vect\theta^\star}(\vect f(\vect s,\vect a))+\Psi^N_{\vect\theta^\star}(\vect f(\vect s,\vect a))\nonumber\\
    &\leq \ell(\vect s,\vect a)-\lambda_{\vect\theta^\star}(\vect f(\vect s,\vect a))+\Psi^{N-1}_{\vect\theta^\star}(\vect f(\vect s,\vect a))
\end{align}
where we used \eqref{eq:mono} in the last inequality. From \eqref{eq:QV1} and \eqref{eq:QV3}, we have:
\begin{align}
    \lambda_{\vect\theta^\star} (\vect f(\vect s, \vect a)) - \lambda_{\vect\theta^\star} (\vect s) 
    &\leq -\hat \ell_{\vect\theta^\star}(\vect s, \vect a)+ \ell(\vect s, \vect a) \nonumber\\&\leq -\alpha_1(\|\vect s-\vect s_e\|)+\ell(\vect s, \vect a)\qedhere
\end{align}
\end{proof}
In order to satisfy \eqref{eq:hat:l}, we generalize it to hold for all parameters set $\vect\theta$, i.e:
\begin{align}
       \alpha_1(\|\vect s-\vect s_e\|)&\leq \hat \ell_{\vect\theta}(\vect s,\vect a),\quad \forall (\vect s,\vect a)\in\mathbb{Z}\label{eq:hat:l}
\end{align}
Under assumption \ref{assum:T}, if the inequality \eqref{eq:hat:l} holds, MPC-scheme \eqref{eq:V} is stable (\cite{MPCbook}), i.e. a dissipative problem has an equivalent tracking MPC which is stable. Then the policies converge to the cost-free steady-state point ($\ell(\vect s_e,\vect a_e)=0$) as time goes to infinity. Under assumption \ref{assum1} and \ref{assum:rich}, the undiscounted Q-learning \eqref{eq:Qlearning} converges to its fixed-point for a dissipative problem in the set where the action-value function is bounded~(see \cite{tsitsiklis1994asynchronous}).
From Theorem \ref{theorem}, the inequality \eqref{eq:hat:l} yields the dissipativity inequality \eqref{eq:dissip:l}. Unfortunately, Q-learning steps delivered by \eqref{eq:Qlearning} do not necessarily respect
the requirement of the inequality \eqref{eq:hat:l} for stage cost $\hat \ell_{\vect\theta}$. We discuss briefly here two solutions to satisfy \eqref{eq:hat:l} and each solution will be extended in future works.

1) Parametrize the stage cost $\hat \ell_{\vect\theta}$ such that be zero at steady-state and strictly positive otherwise using a positive output function approximator. E.g., for a quadratic approximator, one can use the following function approximator:
\begin{align}
    \hat \ell_{\vect\theta}(\vect s,\vect a)=&[(\vect s-\vect s_e)^\top,(\vect a-\vect a_e)^\top]\big(\vect M(\vect \theta)\vect M^\top(\vect \theta)+\nonumber\\&\qquad\qquad\qquad\epsilon \vect I\big)[(\vect s-\vect s_e)^\top,(\vect a-\vect a_e)^\top]^\top
\end{align}
where $\vect M$ is an arbitrary parameters matrix and $\epsilon$ is a small positive constant. In order to provide a more generic positive stage cost, one can use a positive activation function in the output layer of a neural network approximator.

2) Using SOS method for the polynomial approximator of the stage cost $\hat \ell_{\vect\theta}$. In this case, the positivity of the stage cost can be represented as the following linear matrix inequality:
\begin{align}\label{eq:R}
    0\prec  \vect R(\vect \theta_{k}), \quad \forall k\geq0
\end{align}
where
\begin{align}
    \hat \ell_{\vect\theta}(\vect s,\vect a)=\vect m^\top(\vect s-\vect s_e,\vect a-\vect a_e)\vect R(\vect \theta_{k})\vect m(\vect s-\vect s_e,\vect a-\vect a_e)
\end{align}
and where $\vect m$ is a vector of monimials of $\vect s-\vect s_e$ and $\vect a-\vect a_e$ without bias (constant) value to preserve to be zero at the steady-state point. In this case, the update rule \eqref{eq:Qlearning} can be seen as minimizing a quadratic cost subject to some constraints in the parameters taking the form of a Semi-definite program (SDP):
\begin{subequations}\label{eq:const:Q:learning}
\begin{align}
    \min_{\Delta\vect\theta}&\quad \frac{1}{2} \|\Delta\vect\theta\|^2-\alpha\delta_{k}\nabla^\top_{\vect \theta}Q_{\vect \theta}(\vect s_{k},\vect a_k)\Delta\vect\theta\\
    \mathrm{s.t.} &\quad 0\prec \vect R(\Delta\vect\theta+\vect\theta_k) 
\end{align}
\end{subequations}
where $\Delta\vect\theta:=\vect\theta_{k+1}-\vect\theta_k$. One can observe that $\vect\theta_k$ will satisfy \eqref{eq:R}.
Note that using SOS method in this method has less complexity and more accuracy in the storage function rather than solving \eqref{eq:dissip:l} directly using SOS method.
Indeed, to solve dissipativity \eqref{eq:dissip:l} using SOS one needs to approximate the exact stage cost $\ell$ and dynamics by polynomials and the obtained storage function is based of these approximations and may have a bias with the exact storage function (\cite{pirkelmann2019approximate}). However, to solve \eqref{eq:const:Q:learning}, one only needs to provide a SOS MPC stage cost $\hat \ell_{\vect\theta}$. Using a generic approximator for the storage function and terminal cost, Q-learning will be able to find the best polynomial stage cost and capture the optimal action-value function. Indeed, the storage function is used as a generic approximator, and $\ell$ and dynamics do not need to be approximated by the polynomial.
\par In practice, if the parametrization is not rich (assumption \ref{assum:rich} does not hold), $\delta$ will not converge to zero necessarily. In this case, Q-learning will find the optimal parameters among the provided functions family. Then we will check the dissipativity inequality \eqref{eq:dissip:l} for the optimal parameters. If it holds, then the problem is dissipative, otherwise the algorithm is inconclusive, i.e. the parametrization is not rich enough or the problem is non-dissipative. The proposed approach is summarized in the Algorithm \ref{Alg}.
\begin{algorithm2e}[ht!]
	\caption{Q-learning to evaluate storage function and verify the dissipativity}
	\label{Alg}
	\SetKwInOut{Input}{Input}
	\Input{Parameterize $\lambda_{\vect\theta}$, $T_{\vect\theta}$, and $\hat \ell_{\vect\theta}$ \\
	Initialize $\vect\theta_0$}
	\While{$\|\delta\|>\epsilon$}{
		Update $\vect\theta$ from \eqref{eq:const:Q:learning} if $\hat \ell_{\vect\theta}$ is SOS (or from \eqref{eq:Qlearning} if $\hat \ell_{\vect\theta}$ is positive definite by construction) \\}
		\If{Converge}{Compute rotated storage function $\bar\ell_{\vect\theta}$ from learned storage function in \eqref{eq:rotated:L:theta} \\
			\If{$\bar\ell_{\vect\theta}>\epsilon\|\vect s-\vect s_e\|^2$ for small enough $\epsilon>0$}
			{System is dissipative and $\lambda_{\vect\theta}$ is a valid storage function}
			\Else{Inconclusive}}
		\Else{Inconclusive}
\end{algorithm2e} 
\section{Simulation}\label{sec:sim}
In this section, we provide four numerical examples in order to illustrate the efficiency of the proposed method.
\subsection{LQR case}
Firs, we look at a trivial Linear dynamics, Quadratic stage cost, regulator (LQR) problem with non-positive stage cost. LQR is a well-known problem, because the exact optimal policy, value function and stage cost are obtainable using other techniques, e.g., Riccati equation. Consider the following linear dynamics with a quadratic economic stage cost:
\begin{equation}
s_{k+1}=2s_{k}-a_{k}\,,\quad L(s,a)=s^{2}+a^{2}+4sa \ngeqslant 0\label{eq}
\end{equation}
The optimal steady state is $(s_{s},a_{s})=(0,0)$. Using Riccati equation, the optimal value function and policy read as:
\begin{align}
    V^\star(s)=11.7446s^2\,,\quad \pi^\star(s)=1.6861s
\end{align}
The storage function $\lambda_{\vect\theta}$, terminal cost $T_{\vect\theta}$ and stage cost approximations $\hat\ell_{\vect\theta}$ are selected as follows: 
\begin{subequations}
\begin{align}\label{eq:MPC:para:sim}
\lambda_{\vect\theta} (s)&=\theta_1s^2 \,,\quad
T_{\vect\theta}(s)= \theta_2 s^2\\
\hat\ell_{\vect\theta}(s)&=\begin{bmatrix}
s\\a 
\end{bmatrix}^\top \begin{bmatrix}
\theta_3 &\theta_4\\ 
\theta_5&\theta_6
\end{bmatrix}
\begin{bmatrix}
s\\a 
\end{bmatrix}
\end{align}
\end{subequations}
where $\vect\theta=\{\theta_1,\ldots,\theta_6\}$ is the set of  parameters adjusted by Q-learning. The RL steps are restricted to making $T_{\vect\theta}$ positive definite. We initialize $\vect\theta_0=[0.1,1,1,0,0,0]^\top$. Fig. \ref{fig:1} shows the convergence of the parameters resulting from Q-learning during $50$ episodes. As can be seen in Fig.\ref{fig:2}, after $50$ iterations Q-learning is able to capture the optimal value and optimal policy functions.
\begin{figure}[ht!]
\centering
\includegraphics[width=0.48\textwidth]{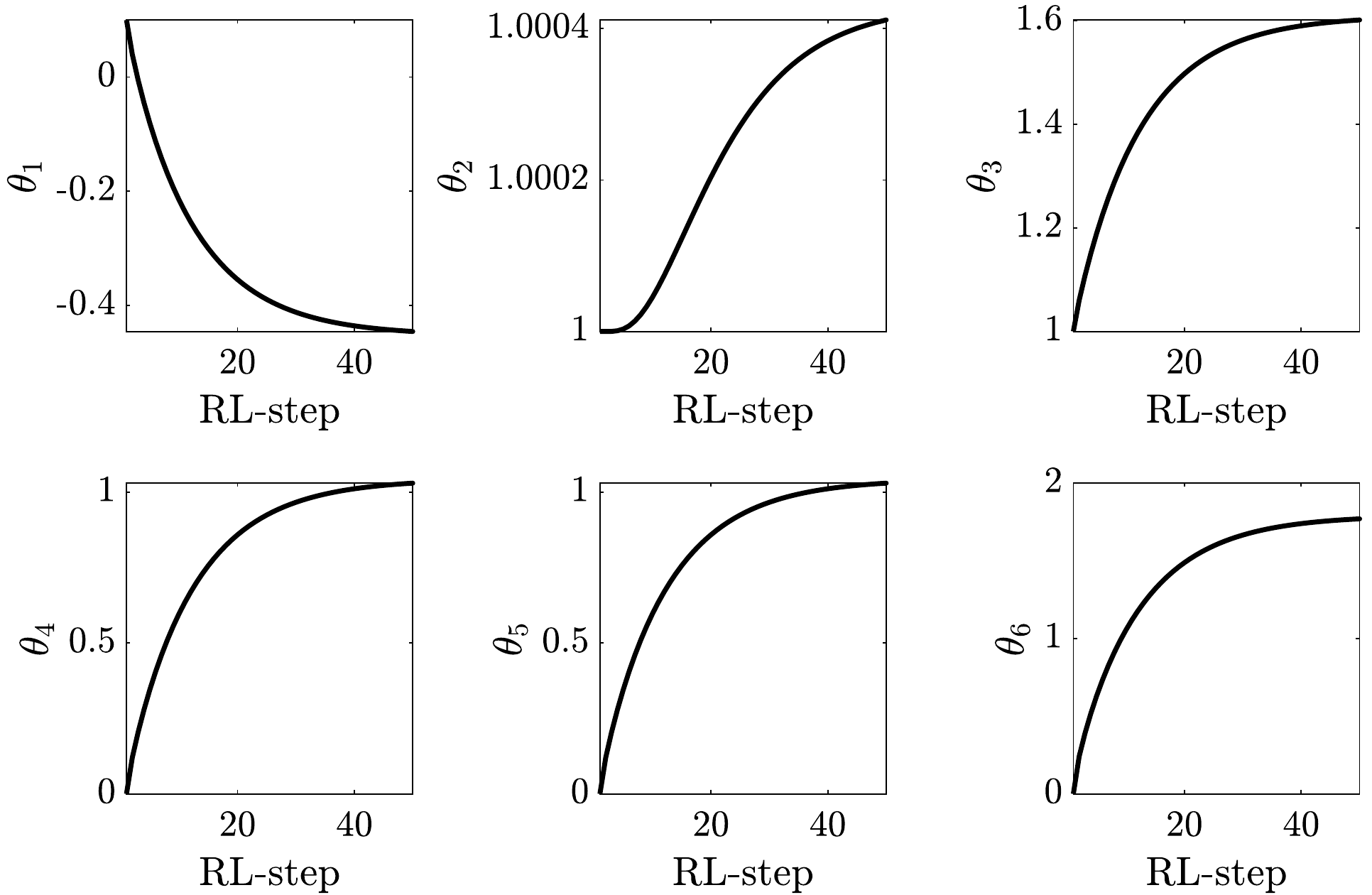}
\caption{Parameters update by Q-learning}
\label{fig:1}
\end{figure}
\begin{figure}[ht!]
\centering
\includegraphics[width=0.48\textwidth]{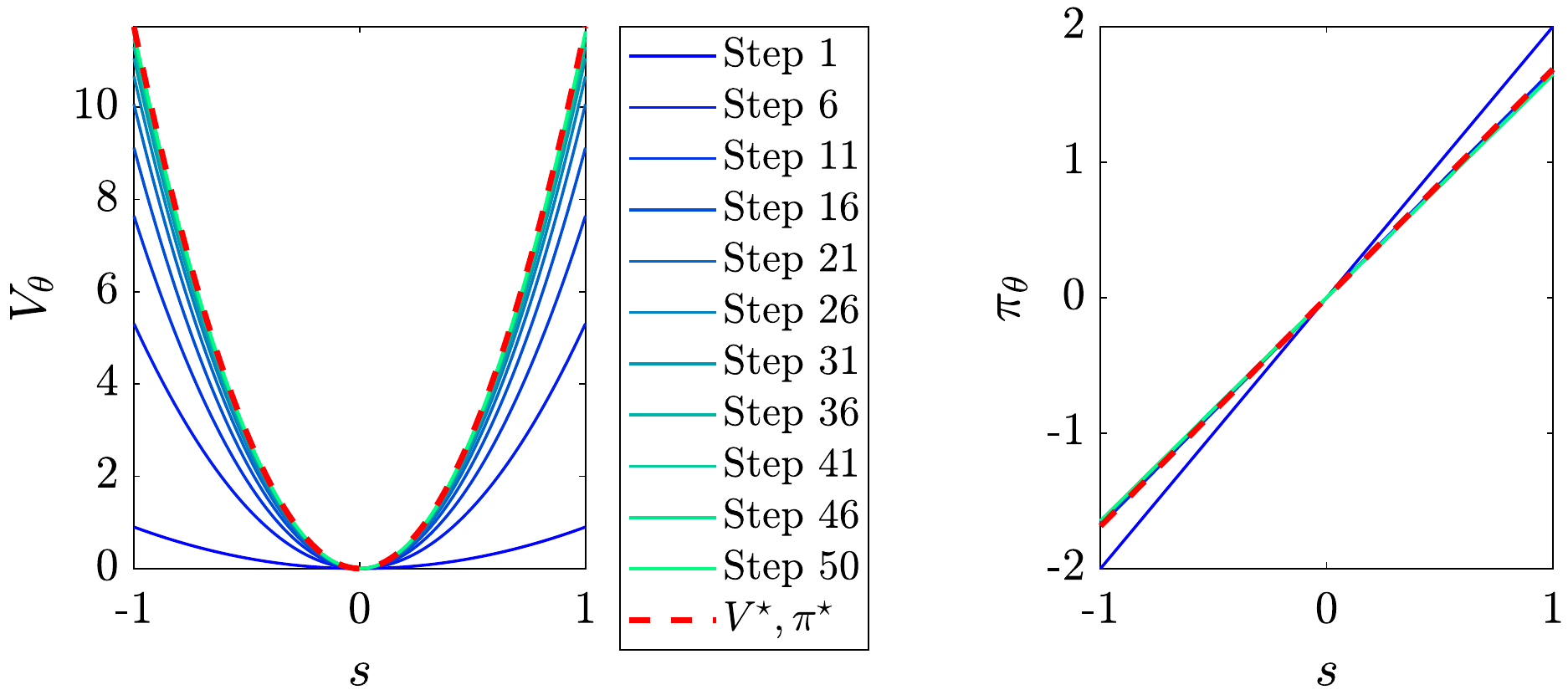}
\caption{(Left) Value function and (Right) Policy function during the learning}
\label{fig:2}
\end{figure}
The learned storage function is $\lambda_{\vect\theta}=-0.4456s^2$. From \eqref{eq:rotated:L:theta} the learned rotated stage cost satisfies the dissipativity inequality:
\begin{align}
  \bar\ell_{\vect\theta}(s,a)=s^2+ a^2+4sa &-0.4456 s^2 \nonumber\\ &+0.4456 (2s- a)^2\geq\rho s^2
\end{align}
for $0<\rho\leq2.3286$.
\subsection{Non-dissipative dynamics}\label{nondis:ex}
We provide next an example of non-dissipative dynamics and stage cost. Consider the following dynamics:
\begin{align}\label{eq:ex:2}
        s_{k+1}=a_k \,,\quad     L(s,a)=-2s^2+a^2+sa+s^2a^2
\end{align}
the optimal steady-state is  $s_e=a_e=0$. In the Appendix, it is shown that there is no storage function for this example, i.e., \eqref{eq:ex:2} is non-dissipative. A quadratic stage cost, terminal cost, and storage function with adjustable parameters are used for the simulation. Fig. \ref{fig:22} shows the learned rotated sage cost after convergence. It can be seen that  Q-learning does not manage to learn a positive definite rotated stage cost.
\begin{figure}[ht!]
\centering
\includegraphics[width=0.48\textwidth]{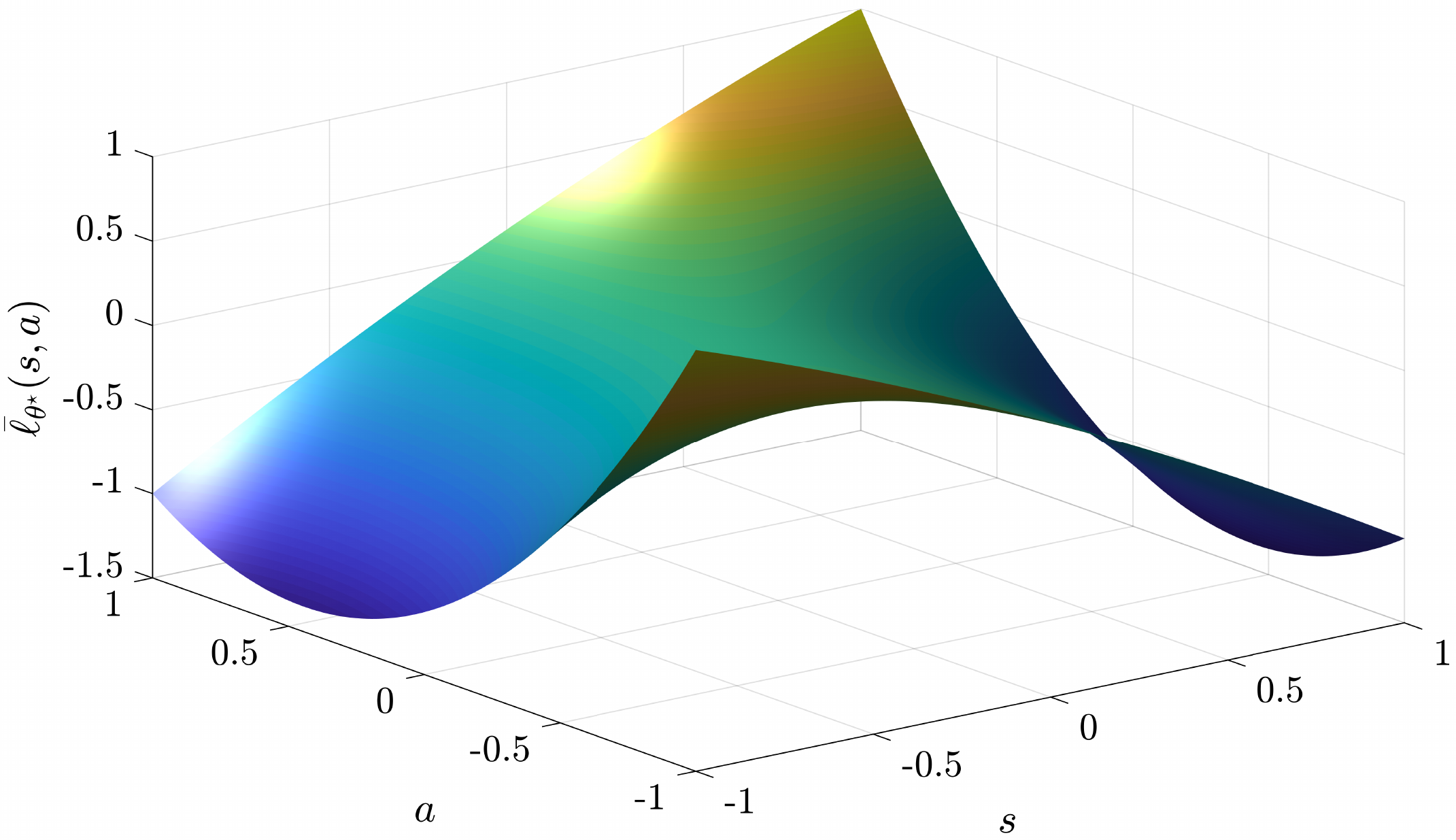}
\caption{Rotated stage cost according to the learned storage function}
\label{fig:22}
\end{figure}
\subsection{Non-polynomial case}
Consider the following dynamics with a non-polynomial economic stage cost:
\begin{align}\label{eq:non:pol}
    s_{k+1}=a_k \,,L(s,a)=-\ln(5s^{0.34}-a),\, 
\end{align}
and
\begin{align}
    \mathbb{X}=[0,10],\,\mathbb{U}=[0.01,10]
\end{align}
This model is a benchmark optimal investment problem where $s$ denotes the investment in a company and the term $5s^{0.34}$ is the return from this investment after one period. Then $5s^{0.34}-a$ is the amount of money that can be used for consumption in the current time period. Then the objective is to maximize the sum of logarithmic utility function. The optimal steady-state point is $s_e=a_e=2.2344$. It is shown in \cite{pirkelmann2019approximate} that a storage function in the form
\begin{align}
   \lambda_{\theta}(s)=\theta(s-s_e),
\end{align}
is valid for this problem and the Sum-of-Squares (SOS) approach delivers an approximated storage function ($\theta=0.23$) using an order-3 Taylor approximation for the stage cost. The analytical solution can be obtained for $\theta$ (see~\cite{pirkelmann2019approximate}). We use the approximated storage function obtained from SOS as an initial guess for $\lambda_{\vect\theta}$ and apply the proposed learning method. The following stage cost is used in the simulation:
\begin{align}
\hat\ell_{\theta}(s,a) = & \bar\ell_{\theta}(s,a) =\\&L(s,a)-L(s_e,a_e)+\theta(s-s_e)-\theta(a-s_e),\nonumber
\end{align}
and we use a long horizon $N=100$ without terminal cost. Fig. \ref{fig:25} shows the update of the parameter $\theta$ over the $100$ episodes. As can be seen, the learning-based storage function converges to the analytical solution, while the parameter $\theta$ resulting from SOS method has a bias. Note that this example is simple and the bias issue is not significant, but for a more complex problem in practice the SOS method may lead a storage function that has more bias with the true storage function. Fig. \ref{fig:256} illustrates the rotated stage cost from the learned storage function.
\begin{figure}[ht!]
\centering
\includegraphics[width=0.48\textwidth]{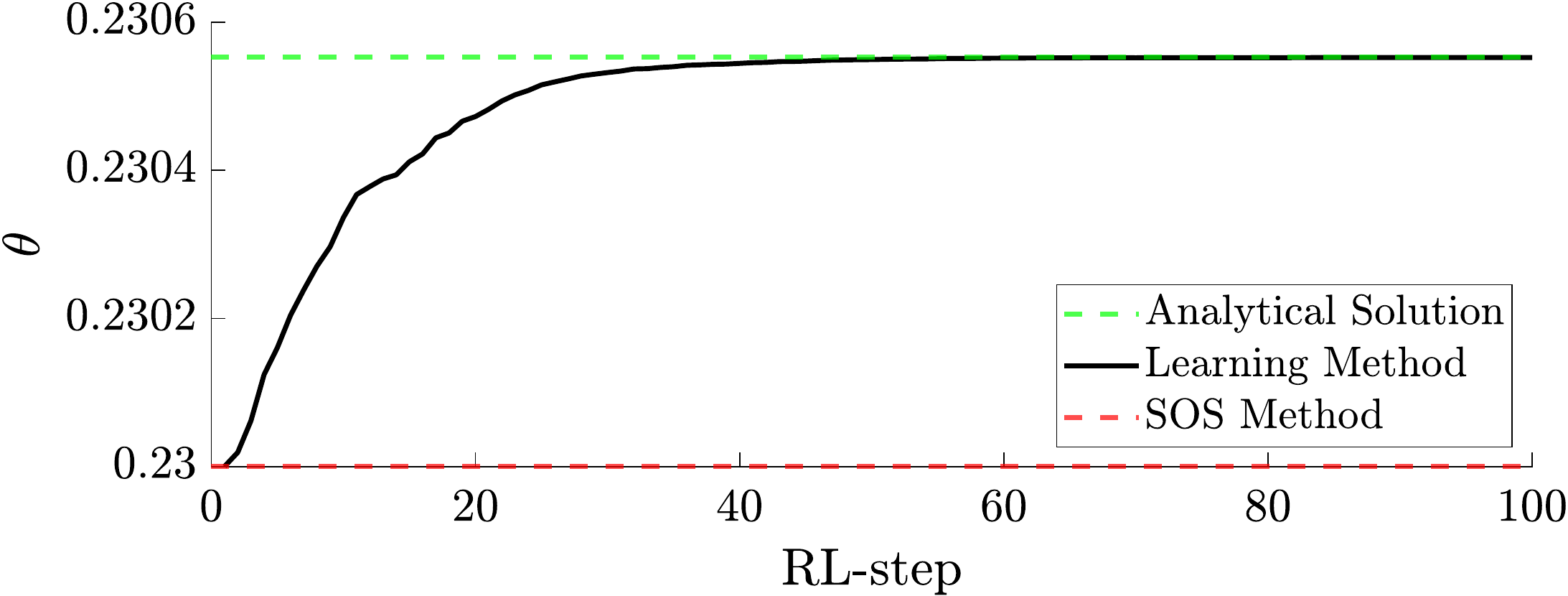}
\caption{Learning of the storage function parameter}
\label{fig:25}
\end{figure}
\begin{figure}[ht!]
\centering
\includegraphics[width=0.48\textwidth]{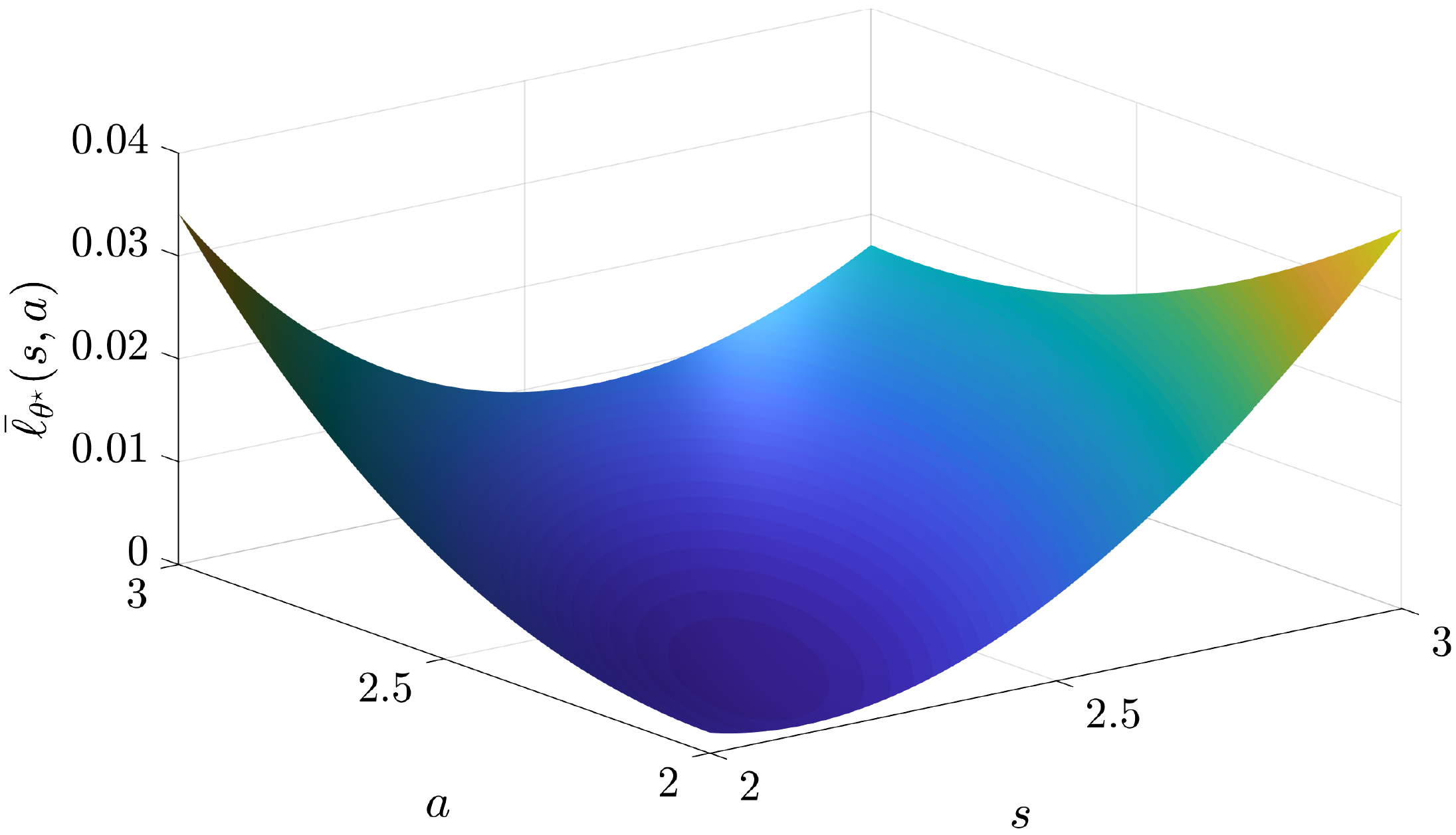}
\caption{Rotated stage cost according to the learned storage function}
\label{fig:256}
\end{figure}
\section{CONCLUSION}\label{sec:conc}
This paper presented the use of undiscounted Q-learning to evaluate the storage function and verify the dissipativity in the ENMPC-scheme for general discrete-time dynamics and cost subject to input-state constraint. We parametrized the equivalent tracking MPC-scheme and delivered it to the RL as a function approximator. We showed that if the parameterization is rich enough, Q-learning will be able to deliver a valid storage function for a dissipative problem. The efficiency of the method was illustrated in the different case studies. Combining this method with SOS, applying for the noisy data and stochastic systems can be considered in future research.
\bibliography{nmpc}   
\appendix
    \section*{Appendix. Non-dissipativity of the Example \ref{nondis:ex}}\label{appen}
Let us assume that the system-stage cost described in \eqref{eq:ex:2} is dissipative. Then there exists a storage function $\lambda(s)$ and $0<\rho$ that satisfy the following inequality:
\setcounter{equation}{0}
\renewcommand{\theequation}{A.\arabic{equation}}
\begin{equation}\label{eq:app1}
        \lambda(a)-\lambda(s)\leq -2s^2+a^2+sa+s^2a^2-\rho s^2,
\end{equation}
and  $\lambda(0)=0$. For $s=0$ we have:
\begin{align}\label{eq:cont:1}
    \lambda(a)\leq a^2 \Rightarrow \lambda(s)\leq s^2,
\end{align}
where we changed the variable name $a\rightarrow s$ in the last inequality. For $a=0$, \eqref{eq:app1} reads as:
\begin{align}\label{eq:cont:2}
    -\lambda(s)\leq -2s^2-\rho s^2 \xRightarrow{0<\rho} 2s^2\leq\lambda(s),
\end{align}
the contradiction from \eqref{eq:cont:1} and \eqref{eq:cont:2} shows that the storage function does not exist and the system is non-dissipative.
\end{document}